\documentclass{article}

\usepackage[utf8]{inputenc}
\usepackage{amsmath,amssymb,amsfonts,amsthm}
\usepackage{algorithm,algpseudocode}
\usepackage{graphicx}
\usepackage[margin=1in]{geometry}
\usepackage[authoryear]{natbib}
\usepackage{hyperref}
\usepackage{booktabs}
\usepackage{caption}
\usepackage{subcaption}
\usepackage{multirow}

\newtheorem{proposition}{Proposition}

\newtheorem{corollary}{Corollary}
\newtheorem{definition}{Definition}

\title{\textbf{Renet: Principled and Efficient Relaxation for the Elastic Net via Dynamic Objective Selection}}

\author{Albert Dorador\\
Universitat Politècnica de Catalunya\\
\texttt{albert.dorador@upc.edu}}

\begin{document}

\maketitle

\begin{abstract}
We introduce Renet, a principled generalization of the Relaxed Lasso to the Elastic Net family of estimators. While, on the one hand, $\ell_1$-regularization is a standard tool for variable selection in high-dimensional regimes and, on the other hand, the $\ell_2$ penalty provides stability and solution uniqueness through strict convexity, the standard Elastic Net nevertheless suffers from shrinkage bias that frequently yields suboptimal prediction accuracy.  We propose to address this limitation through a framework called \textit{relaxation}. Existing relaxation implementations rely on naive linear interpolations of penalized and unpenalized solutions, which ignore the non-linear geometry that characterizes the entire regularization path and risk violating the Karush-Kuhn-Tucker conditions. Renet addresses these limitations by enforcing sign consistency through an adaptive relaxation procedure that dynamically dispatches between convex blending and efficient sub-path refitting. Furthermore, we identify and formalize a unique synergy between relaxation and the ``One-Standard-Error''  rule: relaxation serves as a robust debiasing mechanism, allowing practitioners to leverage the parsimony of the 1-SE rule without the traditional loss in predictive fidelity.  Our theoretical framework incorporates automated stability safeguards for ultra-high dimensional regimes and is supported by a comprehensive benchmarking suite across 20 synthetic and real-world datasets, demonstrating that Renet consistently outperforms the standard Elastic Net and provides a more robust alternative to the Adaptive Elastic Net in high-dimensional, low signal-to-noise ratio and high-multicollinearity regimes. By leveraging an adaptive solver backend, Renet delivers these statistical gains while offering a computational profile that remains competitive with state-of-the-art coordinate descent implementations.
\end{abstract}

\section{Introduction}

\subsection{Motivation}
In high-dimensional regression problems where the number of features $p$ may exceed the number of observations $n$, the Elastic Net \citep{Zou2005} has become a standard tool. By combining the variable selection properties of the Lasso ($\ell_1$ penalty) with the stabilizing (``grouping") effect and strict convexity of Ridge ($\ell_2$ penalty), Elastic Net offers a robust alternative to Lasso, particularly in the presence of high multicollinearity.

However, both Lasso and Elastic Net suffer from a fundamental limitation: the regularization required to simultaneously suppress noise and induce sparsity shrinks the coefficients of true signal variables towards zero. This bias can lead to suboptimal prediction performance, especially when the signal-to-noise ratio is high. This issue is compounded in Elastic Net, as it imposes a ``double penalty" on coefficients \citep{Zou2005}, and the prescribed ex-post coefficient re-scaling, while principled, is not always satisfactory.

While the Adaptive Elastic Net (AEN) \citep{Zou2009} attempts to mitigate this bias by utilizing data-driven weights, its performance relies directly on the quality of an initial estimator---a task that is notoriously difficult and prone to weight propagation error with finite samples in the $n < p$ regime or low signal-to-noise environments. Although AEN possesses the ``Oracle property'' under regularity conditions, this theoretical guarantee can fail in finite samples where the initial estimator is unreliable. Instead of focusing on coefficient-specific debiasing through potentially noisy weights, the framework proposed in this work, known as \textit{relaxation}, adopts a more robust global debiasing strategy. By decoupling the variable selection process from the final coefficient estimation, relaxation allows for a sparse model that retains the predictive accuracy of unconstrained solutions without sacrificing the sparsity and numerical stability provided by the $\ell_1$ and $\ell_2$ penalties, respectively.

Furthermore, while AEN provides a single point-estimate correction based on a fixed weighting scheme that can be satisfactory in well-behaved regimes, it lacks a mechanism to explore the bias-variance tradeoff within a fixed support. In contrast, introducing a relaxation parameter $\theta$ allows for a continuous navigation of the estimation space.  For each candidate sparse model, a Relaxed Elastic Net can reduce shrinkage bias along a continuous path. This allows the estimator to reach optimal risk profiles that are structurally inaccessible to AEN, which must rely on the accuracy of its local weights and is unable to ``re-tune'' the shrinkage intensity without altering the weighted penalty structure itself.

On the other hand, when defined ($n \leq p$), unpenalized Ordinary Least Squares (OLS) linear regression, while not suffering from coefficient bias under standard assumptions, often exhibits a higher mean squared prediction error than Elastic Net due to its increased variance, besides providing zero sparsity. If stepwise variable selection is used in an attempt to obtain a sparse solution, the estimated coefficients are no longer unbiased but tend to be overestimated instead. Furthermore, stepwise selection is a greedy heuristic that scales poorly with dimensionality, as it requires $\mathcal{O}(p^2)$ OLS fits.

Allowing the Elastic Net to optionally debias coefficient estimates---whether partially or in full---emerges as a promising middle ground between the standard, often over-regularized Elastic Net and the less stable or computationally inefficient stepwise OLS linear regression. This approach, which we term Renet, utilizes a framework that is overall more robust in finite samples than that of AEN, providing a self-correcting path that balances parsimony with predictive fidelity even in the absence of the  regularity conditions that are required for asymptotically valid individual coefficient debiasing.

\subsection{The Relaxation Problem}
To mitigate the estimation bias in $\ell_1$-regularized linear regression, \cite{Meinshausen2007} introduced the Relaxed Lasso, which applies the variable selection of Lasso but re-estimates the coefficients with a reduced penalty. 

However, as successful as the Lasso has been in many situations, it has a few well-known design flaws that may hinder performance:
\begin{itemize}
\item Saturation: in the $n < p$ regime, the Lasso is only able to select at most $n$ variables, which imposes a hard constraint that is driven by the mathematical construction of the estimator, not by the data
\item Spurious selection in groups of highly correlated predictors, which undermines model stability and may compromise interpretation
\item Under high multicollinearity, the Ridge estimator often exhibits superior performance \citep{Tibshirani1996}
\end{itemize}

All of these shortcomings are theoretically solved  by introducing an $\ell_2$ penalty term, giving rise to the Elastic Net estimator (although it does introduce a complication in a relaxation framework, see Section \ref{sec:stability}). Given a dataset ($\mathbf{X}, \mathbf{y}$) where $\mathbf{X} \in \mathbb{R}^{n \times p}$ and $\mathbf{y} \in \mathbb{R}^{n \times 1}$, the Elastic Net estimator $\hat{\beta}$ is defined as the minimizer of the following objective function:
\begin{equation} \label{eq:EN}
    \mathcal{L}(\beta; \lambda, \alpha) = \frac{1}{2n} \| \mathbf{y} - \mathbf{X}\beta \|_2^2 + \lambda \left( \alpha \| \beta \|_1 + \frac{1-\alpha}{2} \| \beta \|_2^2 \right)
\end{equation}
where $\lambda \geq 0$ controls the overall strength of the penalty, and $\alpha \in [0, 1]$ governs the balance between the Lasso and Ridge components.

While relaxation is theoretically sound, there are not many software implementations of Relaxed Lasso, let alone Relaxed Elastic Net.  The R package \texttt{glmnet} \citep{Friedman2010} provides one of the few implementations of the Relaxed Elastic Net publicly available -- perhaps the only one. However, based on its documentation, it does not follow the relaxation algorithm in \cite{Meinshausen2007} but instead implements relaxation as a simple linear interpolation between the Elastic Net solution and the OLS solution on the active set:
\begin{equation}
    \hat{\beta}_{naive} = \theta \hat{\beta}_{EN} + (1-\theta) \hat{\beta}_{OLS}
\end{equation}
This approach is computationally inexpensive but mathematically compromised when the regularization path is non-linear,  and, specifically, when relaxing the penalty causes variables to change sign or re-enter the active set. This approach, thus, may yield a model that never actually satisfied the Karush-Kuhn-Tucker (KKT) conditions of the original objective function at any $\lambda \geq 0$ (unless one of the extreme points is chosen). As already warned in \cite{Tibshirani1996} and later proved formally (see e.g. \cite{Dorador2025}), the size of the estimated regression coefficients in Lasso is in general not a monotonic function of the penalty parameter,  as it depends on the correlation structure present in the design matrix. This means that the linearity violation in the regularization path is far from an edge case. In addition, under $n < p$, the OLS solution is not unique and highly unstable, which suggests that mixing it with the Elastic Net solution as proposed by \cite{Friedman2010} can be dangerous.  Indeed, while in practice this simplified approach can be relatively harmless in the case of Relaxed Lasso, this is not the case for the Relaxed Elastic Net with a nonzero $\ell_2$ penalty -- as alluded earlier -- which will be discussed in greater detail in Section \ref{sec:stability}.

Having said that, under coefficient sign consistency between the penalized and unpenalized estimates (which is not uncommon in well-behaved regimes), our proposed relaxation implementation does coincide with the approach followed in \texttt{glmnet}. In more general scenarios where sign consistency does not hold, our approach provides the necessary mathematical safeguards to ensure KKT compliance.

Since our implementation generalizes the approach followed in \texttt{glmnet}, we have excluded  the latter from our benchmarks. This has the additional benefit of allowing a pure Python benchmarking, which avoids spurious environment-related differences that may confound results.

\subsection{Our Contribution: Renet}
\label{sec:contrib}
We propose Renet, a principled, strictly convex, stability-enhanced implementation of the Relaxed Elastic Net. Our main contributions are four-fold:
\begin{enumerate}
    \item \textbf{Generalization and convexity:} We generalize Meinshausen's framework to the Elastic Net, ensuring strict convexity of the objective function by introducing an $\ell_2$ penalty. This guarantees uniqueness of the solution, resolving identifiability issues inherent in pure Relaxed Lasso.
    \item \textbf{Principled objective selection and stability:} Unlike naive implementations, Renet dynamically switches between an unpenalized objective that allows for fast analytical interpolation (when coefficient signs are stable) and sub-path refitting (when geometry dictates), remaining stable even when $n \ll p$ by virtue of our theoretically-grounded complexity-adjusted relaxation floor. Whenever $n < p_{active}$ (due to the effect of the $\ell_2$ penalty),  relaxation is avoided entirely to preserve stability.
    \item \textbf{Solver-agnostic:} While Meinshausen's ``Refined Algorithm" requires a LARS \citep{Efron2004} implementation, Renet does not: the gradients from a LARS step are no longer necessary. This means Renet is compatible with any general-purpose optimizer like e.g. coordinate descent, widening its applicability and allowing to tackle the $n < p$ regime more efficiently, as we will discuss in Section \ref{sec:Agnost}.
    \item \textbf{Synergy with the ``One-Standard-Error'' (1-SE) rule:} We provide a theoretical rationale for why relaxation uniquely complements the well-known 1-SE rule, keeping its parsimony-inducing properties while greatly mitigating the usual accuracy-parsimony tradeoff by correcting the shrinkage bias inherent in over-regularized models.
\end{enumerate}

\noindent The remainder of this article is organized as follows. Section \ref{sec:Renet} details the Renet algorithm, discussing its solver-agnostic nature and stability safeguards in high-dimensional settings. Section \ref{sec:Theory} explores the theoretical properties of the estimator, including its grouping effects, bias-mitigation factors in orthogonal designs, as well as a formal justification for the synergy between relaxation and the 1-SE rule. Section \ref{sec:Empirics} presents our empirical evaluation on synthetic and real-world benchmarks, with concluding remarks in Section \ref{sec:conclusion}.

\section{The Renet Algorithm}
\label{sec:Renet}

The Renet estimator defines a two-stage procedure designed to decouple the variable selection process from the coefficient estimation. Let $\mathbf{X} \in \mathbb{R}^{n \times p}$ be the design matrix and $\mathbf{y} \in \mathbb{R}^{n \times 1}$ the response vector.

\paragraph{Stage 1: Feature Selection}
In the first stage, we identify a candidate active set $\mathcal{A}_{\lambda}$ by solving the standard Elastic Net objective:
\begin{equation} \label{eq:EN_Selection}
    \hat{\beta}_{EN}(\lambda, \alpha) = \underset{\beta \in \mathbb{R}^{p}}{\text{arg min}} \left\{ \frac{1}{2n} \| \mathbf{y} - \mathbf{X}\beta \|_2^2 + \lambda \left( \alpha \| \beta \|_1 + \frac{1-\alpha}{2} \| \beta \|_2^2 \right) \right\}
\end{equation}
The active set is defined as the indices of the non-zero coefficients: $\mathcal{A}_{\lambda} = \{j : \hat{\beta}_{EN,j} \neq 0\}$.

\paragraph{Stage 2: Adaptive Relaxation}
In the second stage, we refit the model using only the features indexed by $\mathcal{A}_{\lambda}$. To mitigate the shrinkage bias inherent in the initial selection, we solve a restricted objective with a scaled penalty $\theta \lambda$ (the Renet objective), as we show next.

\begin{definition}[Renet Estimator] \label{def:Renet}
Let $\mathbf{X}_{\mathcal{A}_{\lambda}}$ be the sub-matrix of $\mathbf{X}$ containing only columns indexed by the active set. For a relaxation parameter $\theta \in (0, 1]$, the Renet estimator is defined as:
\begin{equation} \label{eq:RenetSol}
    \hat{\beta}_{Renet}(\lambda, \alpha, \theta) = \underset{\beta \in \mathbb{R}^{|\mathcal{A}_{\lambda}|}}{\text{arg min}} \left\{ \frac{1}{2n} \| \mathbf{y} - \mathbf{X}_{\mathcal{A}_{\lambda}}\beta \|_2^2 + \theta \lambda \left( \alpha \| \beta \|_1 + \frac{1-\alpha}{2} \| \beta \|_2^2 \right) \right\}
\end{equation}
\end{definition}

\noindent A Python implementation of Renet is included in the \texttt{trust-free} package (version $\geq$ 3.0.0), both as a standalone estimator (\texttt{Renet} class) and as an estimator in the leaves of a TRUST \citep{Dorador2026TRUST} tree (within the \texttt{TRUSTRegressor} class).  Pseudocode is provided in Appendix \ref{ap:Algo}.

\subsection{Solver-Agnosticism}
\label{sec:Agnost}

Renet implements a solver-agnostic variation of Meinshausen's ``Refined Algorithm''.  While \cite{Meinshausen2007} relies on differential path updates specific to the LARS algorithm  to predict sign violations, Renet explicitly computes the unregularized (or, for stability,  minimally $\ell_2$-penalized) endpoint on the active set. If the sign consistency condition holds, i.e., $\text{sign}(\hat{\beta}_{EN}) = \text{sign}(\hat{\beta}_{OLS})$, the solution lies on a linear segment and Renet utilizes efficient convex blending, which is mathematically equivalent to Meinshausen's linear interpolation in the sign-consistent regime. If signs conflict (indicating a variable crossing zero), the linearity assumption held by naive relaxation implementations fails and Renet correctly reverts to sub-path refitting, ensuring the solution respects the non-linear geometry of the KKT conditions even when the regularization path is non-linear. 

Beyond wider compatibility, making Renet model-agnostic has profound implications, especially in the high-dimensional regime, as we discuss next.

\subsection{Computational Advantage over LARS-based Implementations}

The original Elastic Net formulation \citep{Zou2005} suggests a LARS-based approach which requires an augmented design matrix $\mathbf{X}^* \in \mathbb{R}^{(n+p) \times p}$. This dimensionality expansion is inexpensive for small problems, but, as the authors themselves acknowledge,  can be computationally prohibitive in the $n \ll p$ regime. Renet sidesteps this entirely. By utilizing the \textit{Celer} solver \citep{Massias2018}, Renet leverages duality gaps to prune the feature space before optimization begins,  which is an efficient strategy in the $n \ll p$ regime. Furthermore, since our algorithm warm-starts the sub-path refit using the initial Elastic Net solution, the second optimization often converges in a fraction of the time required for a cold-start, making the total cost of relaxation modest-to-negligible compared to the initial path computation.

\subsection{Stability Safeguards and the Consensus on the $n \ll p$ Regime}
\label{sec:stability}
In the original Relaxed Lasso \citep{Meinshausen2007}, the $n < p$ regime is implicitly managed by the saturation property of the Lasso, which constrains the active set cardinality to $p_{active} \equiv |\mathcal{A}_{\lambda}| \le n$. Consequently, the relaxation stage in a pure Lasso framework always reduces to a well-defined Least Squares problem. 

However, because the initial Elastic Net estimator can select more than $n$ variables when features are highly correlated if $\alpha < 1$,  Renet may encounter active sets downstream where $n < p_{active}$. In such cases, full relaxation ($\theta \to 0$) would target a non-unique (or highly unstable minimum-norm) solution. 

For that reason,  practical implementations like \texttt{glmnet} \citep{Friedman2010}  have long signaled caution,  attempting to manage this instability through two primary heuristics: \texttt{lambda.min.ratio}, which effectively prevents the estimator from reaching the low-bias, high-variance region near OLS, and the \texttt{maxp} parameter, which imposes a hard limit on the number of relaxed coefficients (defaulting to $n-3$).

We argue that these constraints are somewhat arbitrary and can interfere with the selection process. Renet instead allows the initial screening to explore the entire $\lambda$ spectrum and any number of features $p_{active}$, but then puts two safety measures in place:
\begin{enumerate}
    \item \textbf{No relaxation under saturation:} If the model reaches saturation, i.e., $p_{active} > n$,  relaxation becomes unreliable and is thus avoided by setting $\theta \equiv 1$. This ensures the estimator defaults to the stable penalized solution when the restricted design matrix is not full rank. This, in turn, encourages further sparsity under $n < p$, because only unsaturated models are allowed to debias their coefficients. 
    
    \item \textbf{Complexity-adjusted relaxation floor:} While the original Relaxed Lasso \citep{Meinshausen2007} allows for full relaxation even in ultra high-dimensional regimes, its theoretical guarantees are asymptotic.  To enhance stability (Section \ref{sec:grouping}) and  prevent over-relaxation under high search complexity with finite samples, we define a data-driven lower bound for $\theta$ that applies specifically to the $n \ll p$ regime:
    \begin{equation}
        \theta_{min} = \text{min} \left( 1.0,  \frac{\log p}{2\sqrt{n}} \right)
    \end{equation}
    Observe that this constraint, which can be viewed as a selection bias penalty, vanishes at a rate that is consistent with Assumption 2 in \cite{Meinshausen2007}.
\end{enumerate}

Both these structural constraints create a safe relaxation path that preserves the conditioning benefits of the initial penalty while allowing for significant bias correction even in ultra high-dimensional settings.

\section{Theoretical Properties}
\label{sec:Theory}

\subsection{Generalization of Relaxed Lasso and Elastic Net}
Renet generalizes both the Relaxed Lasso (Relasso) and the ordinary Elastic Net (Enet): setting $\alpha=1$ recovers the Relaxed Lasso of \cite{Meinshausen2007}, while setting $\theta=1$ yields the standard Elastic Net. A critical theoretical advantage of Renet is the inclusion of the $\ell_2$ term ($\alpha < 1$), which imparts strict convexity to the objective function in Equation \ref{eq:EN}. This ensures that for any $\lambda > 0$, the Hessian is positive definite, and the solution $\hat{\beta}$ is unique, avoiding the non-uniqueness pitfalls of Lasso even in the $n <p$ regime.

Let $\mathcal{H}_{Model}$ denote the hypothesis space of a given model class, i.e. the set of all possible parameter estimates reachable by a given family of estimators. The following inclusion chains hold:
\begin{equation}
\mathcal{H}_{Renet} \supseteq \mathcal{H}_{Relasso} \supseteq \mathcal{H}_{Lasso} \supseteq \mathcal{H}_{OLS}
\end{equation}
\begin{equation}
\mathcal{H}_{Renet} \supseteq \mathcal{H}_{Enet} \supseteq \{ \mathcal{H}_{Lasso} \cup \mathcal{H}_{Ridge} \} \supseteq \mathcal{H}_{OLS}
\end{equation}
Consequently, Renet provides a more flexible framework for navigating the bias-variance tradeoff. Provided that the hyperparameter triplet $(\lambda, \alpha, \theta)$ is tuned effectively\footnote{In our subsequent experiments, we fix $\alpha = 0.95$ both to reflect our preference for sparsity and to simplify an already rich experimental setup.} (e.g., via $k$-fold cross-validation), Renet is theoretically guaranteed to perform at least as well as its predecessors. As we will demonstrate in Section \ref{sec:Empirics}, the ability to decouple variable selection (via $\lambda$) from shrinkage intensity (via $\theta$) allows Renet to simultaneously achieve comparable or higher out-of-sample accuracy and sparsity across a broad range of scenarios. The inclusion of a 1-SE rule tends to further improve the accuracy-sparsity combination that Renet provides. This is not a coincidence, as we will formally justify in Section \ref{sec:1SE}.

Beyond this geometric inclusion, Renet offers a powerful Bayesian interpretation that further justifies its two-dimensional parameter space. As noted by \cite{Zou2005}, the Lasso and Ridge penalties correspond to independent Laplace and Gaussian priors on the coefficients, respectively. Standard Elastic Net is forced to use a single parameter $\lambda$ to simultaneously regulate both the sparsity of the support and the intensity of the prior.  This creates a structural conflict: it is mathematically impossible for the standard Elastic Net to represent a model with a very small support and a non-informative prior, or conversely, a dense model with a highly informative prior.

Renet resolves this by decoupling the prior's support from its intensity. In this framework, $\lambda$ acts as a threshold that determines the prior's \textit{support} $\mathcal{A}_{\lambda}$---effectively defining the subspace where the prior has finite variance. Simultaneously, $\theta$ modulates the relative \textit{certainty} of the prior on that subspace. 

Crucially, the Renet objective (Eq. \ref{eq:RenetSol}) does not represent a sequential decision process, but a joint prior configuration. By exploring the $(\lambda, \theta)$ space via cross-validation, Renet identifies the optimal balance between the selection pressure ($\lambda$) and the estimation bias ($\theta$). As $\theta \to 0$, we transition toward a non-informative prior specifically for the active set $\mathcal{A}_{\lambda}$, while maintaining an infinitely sharp (Dirac delta) prior for all variables outside the set. This decoupling allows Renet to leverage the high precision required for consistent selection without being forced to accept the resulting shrinkage bias in the final estimates.

\subsection{Renet as a Stabilized Lasso}
In this section, we show in full generality that Renet can be viewed as a stabilized version of the Lasso (and Relaxed Lasso). By adapting Theorem 2 of \cite{Zou2005}, we demonstrate that the Renet sub-problem is equivalent to a (Relaxed) Lasso problem solved on a modified, more stable geometry.

\begin{proposition}[Lasso Stabilization] \label{prop:LassoStab}
Given a fixed hyperparameter triplet $(\lambda, \alpha, \theta)$ and an active set $\mathcal{A}_{\lambda}$, the Renet estimates $\hat{\beta}_{Renet}$ are the minimizers of:
\begin{equation}
    J(\beta) = \beta^T \left( \frac{\mathbf{X}^T\mathbf{X} + \theta \lambda (1-\alpha)\mathbf{I}}{1 + \theta \lambda (1-\alpha)} \right) \beta - 2\mathbf{y}^T\mathbf{X}\beta + \theta \lambda \alpha \|\beta\|_1
\end{equation}
\end{proposition}

\begin{proof}
The result is a straightforward adaptation of Theorem 2 in \cite{Zou2005}. By substituting the effective Renet penalties $\lambda_{1, \text{relax}} = \theta \lambda \alpha$ and $\lambda_{2, \text{relax}} = \theta \lambda (1-\alpha)$ into the first-order optimality conditions in the Elastic Net and rearranging the quadratic form, we obtain the stated result.
\end{proof}

Similar to the original Elastic Net, Proposition \ref{prop:LassoStab} shows that Renet acts as a (Relaxed) Lasso estimator where the design matrix $\mathbf{X}^T\mathbf{X}$ is ``shrunk'' toward the identity matrix by the $\ell_2$ penalty. However, unlike the standard Elastic Net (where $\theta=1$), Renet explicitly controls the degree of this deformation and the intensity of the $\ell_1$ penalty via the relaxation parameter $\theta$. 

This has a direct advantage. The introduction of the relaxation parameter $\theta$ allows Renet to navigate the bias-variance manifold more effectively than the standard Elastic Net by conducting a joint two-dimensional optimization over the $(\lambda, \theta)$ surface. This ensures that the selection of the regularization intensity $\lambda$ is informed by the debiasing potential of the relaxation parameter $\theta$. This joint optimization identifies optimal models that are structurally inaccessible to the standard Elastic Net. Specifically, this approach allows for the selection of a higher-penalty $\lambda$ to ensure a clean, parsimonious active set, while simultaneously utilizing a lower $\theta$ to recover the signal magnitude lost to shrinkage. Therefore, Renet does not merely ``fix" a biased model (as it would be the case with a naive sequential ``$\lambda$-then-$\theta$" optimization),  it identifies the optimal coordinate on the accuracy-sparsity frontier by treating relaxation as an integral dimension of the model-selection process.

Concretely, as $\theta \to 0$, the quadratic term in the parenthesis recovers the standard OLS geometry $\mathbf{X}^T\mathbf{X}$ and the $\ell_1$ penalty vanishes. This formally proves that Renet navigates the space between the stable, well-conditioned geometry of the Elastic Net and the unbiased but potentially volatile geometry of the OLS. In this framework, $\theta$ serves as a continuous dial for geometric transformation, allowing for a ``relaxed'' fit that remains strictly convex as long as $\theta \lambda (1-\alpha) > 0$.

\subsection{The Renet Grouping Property}
\label{sec:grouping}
A known limitation of the standard Relaxed Lasso is that by removing the penalty (relaxing to OLS), one forfeits the stability benefits gained during the selection phase. Specifically, in the presence of high multicollinearity, OLS coefficients are known to exhibit high variance and poor grouping behavior \citep{Zou2005}. Renet, however, maintains a scaled $\ell_2$ penalty during the relaxation phase. We can thus adapt the ``Grouping Effect'' inequality from \cite{Zou2005} to add further theoretical justification for the complexity-adjusted relaxation floor that we introduced in Section \ref{sec:stability}.

\begin{proposition}[Renet Grouping Effect]
\label{prop:grouping}
Given the Renet hyperparameters $\lambda, \alpha, \theta$, let $\hat{\beta}(\theta)$ be the solution to the Renet sub-problem on the active set. For any pair of (standardized) highly correlated variables $i, j$ in the active set with sample correlation $\rho = \mathbf{x}_i^T \mathbf{x}_j \approx 1$, the difference in their coefficient estimates satisfies:
\begin{equation} \label{eq:grouping}
    |\hat{\beta}_i(\theta) - \hat{\beta}_j(\theta)| \le \frac{1}{\theta \lambda (1-\alpha)} \|\mathbf{y}\|_2 \sqrt{2(1-\rho)}
\end{equation}
\end{proposition}

\begin{proof}
The result follows directly from Theorem 1 in \cite{Zou2005} by substituting the effective regularization strength of the Renet sub-problem. In the relaxation phase, the effective penalty parameters are $\lambda_{relax} = \theta \lambda$. The Ridge component strength is thus $\lambda_2 = \theta \lambda (1-\alpha)$. Substituting this denominator into the original bound yields the stated result.
\end{proof}

An important remark is in order. Equation \ref{eq:grouping} reveals a strong dependency on $\theta$. For fixed $\lambda >0$ and $\alpha < 1$,  as $\theta \to 0$ (approaching the pure OLS solution), the upper bound on the coefficient estimate difference tends to infinity. This shows that full relaxation erodes the grouping structure, allowing arbitrarily large divergences between the estimates of correlated predictors.  Therefore,  this bound highlights a critical flaw in naive relaxation approaches that can only interpolate towards the pure OLS solution.  While Renet does utilize interpolation for computational speed, it does so \textit{conditionally}: only when sign consistency confirms the stability of the path; otherwise, Renet performs a warm-started sub-path refit.  In the ultra high-dimensional regime with finite $n \ll p$, the sub-path refit relaxation is constrained to $\theta \gg 0$ as discussed in Section \ref{sec:stability}, which guarantees that the bound remains tight (provided $\lambda >0$ and $\alpha < 1$), strictly enforcing the grouping effect even in the relaxed phase. 

A similar conclusion holds true when we let $\alpha \to 1$ (Relaxed Lasso limit),  highlighting a clear theoretical advantage of keeping $\alpha < 1$ as opposed to the Relaxed Lasso formulation, which implies $\alpha = 1$.

\subsection{Bias Reduction in Orthogonal Design}
The standard Elastic Net  introduces a well-documented ``double shrinkage'' bias \citep{Zou2005}, which the authors address via a global rescaling factor $1 + \lambda(1-\alpha)$, justified by arguments ranging from ridge operator decomposition to minimax optimality in specific settings. While this rescaling is theoretically principled, it is arguably rigid: it applies a fixed multiplicative correction determined solely by the penalty parameters, effectively assuming that fully reversing the ridge contraction while leaving the Lasso penalty intact is necessarily the optimal strategy for prediction. On the other hand, Renet addresses \textit{both} shrinkage biases at once with a data-driven approach. To explicitly quantify the debiasing effect of Renet, we consider the simplified case of an orthogonal design matrix. In this setting, the Renet solution has a closed-form expression, allowing for a direct analytical comparison of the shrinkage effects.

\begin{proposition}[Bias Mitigation Factor] \label{prop:Bias}
Assume an orthogonal design such that $\mathbf{X}^T\mathbf{X} = n\mathbf{I}$, implying (without loss of generality) the predictors are standardized such that $\mathbf{x_j}^T\mathbf{x_j} = n$ for all $j$. Let $\hat{\beta}_{OLS} = \frac{1}{n}\mathbf{X}^T\mathbf{y}$ be the unconstrained OLS estimator. For a given hyperparameter triplet $(\lambda, \alpha, \theta)$, the Renet estimator $\hat{\beta}_{Renet}$ is given by:
\begin{equation}
    \hat{\beta}_{Renet_j} = \frac{(|\hat{\beta}_{OLS_j}| - \theta \lambda \alpha)_+}{1 + \theta \lambda (1-\alpha)} \text{sgn}(\hat{\beta}_{OLS_j})
\end{equation}
where $(z)_+ = \max(z, 0)$.
\end{proposition}

\begin{proof}
Let $\lambda, \alpha, \theta$ be fixed arbitrarily. The Renet objective function, for the resulting (possibly restricted) feature space $\mathbf{X}$ is then
\begin{equation} \label{eq:RenetObj}
    \mathcal{L}(\beta; \lambda, \alpha, \theta) = \frac{1}{2n} \| \mathbf{y} - \mathbf{X}\beta \|_2^2 + \theta \lambda \left( \alpha \| \beta \|_1 + \frac{1-\alpha}{2} \| \beta \|_2^2 \right)
\end{equation}
Under the assumption of orthogonality, the first term decomposes. Fixing $j$ and carrying out standard algebraic steps we obtain
\begin{equation} \label{eq:RenetObj2}
    \mathcal{L}(\beta_j; \lambda, \alpha, \theta) = \frac{1}{2n} \mathbf{y}^T\mathbf{y} - \frac{\mathbf{x_j}^T\mathbf{y}}{n}\beta_j + \frac{\mathbf{x_j}^T\mathbf{x_j}}{2n}\beta_j^2 + \theta \lambda \alpha |\beta_j| + \frac{\theta \lambda(1-\alpha)}{2}\beta_j^2
\end{equation}
Substituting $\mathbf{x_j}^T\mathbf{x_j} = n$ and $\hat{\beta}_{OLS_j} = \frac{1}{n}\mathbf{x_j}^T\mathbf{y}$, the objective simplifies to:
\begin{equation} \label{eq:RenetObj3}
    \mathcal{L}(\beta_j; \lambda, \alpha, \theta) = \frac{1}{2n} \mathbf{y}^T\mathbf{y} - \hat{\beta}_{OLS_j}\beta_j  + \frac{1}{2}\beta_j^2 + \theta \lambda \alpha |\beta_j| + \frac{\theta \lambda(1-\alpha)}{2}\beta_j^2
\end{equation}

The subgradient optimality condition requires:
\begin{equation}
    (\beta_j - \hat{\beta}_{OLS, j}) + \theta \lambda (1-\alpha)\beta_j + \theta \lambda \alpha s = 0
\end{equation}
where $s \in \partial |\beta_j|$. For $\beta_j \neq 0$, we have $s = \text{sgn}(\beta_j)$. Rearranging for $\beta_j$:
\begin{equation}
    \beta_j [1 + \theta \lambda (1-\alpha)] = \hat{\beta}_{OLS, j} - \theta \lambda \alpha \text{sgn}(\beta_j)
\end{equation}
For a positive solution $\beta_j > 0$, we must have $\hat{\beta}_{OLS, j} > \theta \lambda \alpha$. This leads directly to the soft-thresholding form stated in the proposition.
\end{proof}

Proposition \ref{prop:Bias} reveals a critical insight: by introducing the parameter $\theta$, Renet mitigates \textit{both} fundamental shrinkage biases of the Elastic Net ($\theta=1$) simultaneously. The original \cite{Zou2005} rescaling factor addresses the $\ell_2$ contraction (the denominator) but leaves the $\ell_1$ thresholding (the numerator) untouched. In contrast, under Renet ($\theta < 1$), the threshold drops to $\theta \lambda \alpha$, allowing weaker signals to survive or recover their original magnitude, while the denominator approaches unity, organically reducing the proportional shrinkage toward zero.

Indeed, a crucial distinction between Renet and the standard rescaled Elastic Net \citep{Zou2005} lies in the survival and recovery of small-magnitude signals. The original rescaling factor $1 + \lambda(1-\alpha)$ is a post-hoc correction that acts exclusively on the support already selected at $\theta=1$. Because it is a multiplicative scalar applied after thresholding, it cannot recover variables that were zeroed out by the initial penalty.  In contrast, Renet's relaxation parameter $\theta$ appears inside the soft-thresholding operator $(|\hat{\beta}_{OLS, j}| - \theta \lambda \alpha)_+$. As $\theta$ decreases, the active set expands, enabling Renet to recover signal magnitude for features that would otherwise be strictly suppressed under a standard Elastic Net at the same level of $\lambda$. Consequently, Renet addresses selection bias at its geometric root, whereas post-estimation rescaling merely mitigates ($\ell_2$-induced) estimation bias. 

Consequently, Renet is able to preserve a larger fraction of the original signal, analytically demonstrating how relaxation mitigates estimation bias while the non-zero $(1-\alpha)$ term maintains the strict convexity and uniqueness properties discussed in Section \ref{sec:contrib}. Rather than imposing a fixed correction that targets only the ridge-induced shrinkage, Renet treats the total ``unshrinkage'' intensity as a tunable hyperparameter. This allows the estimator to navigate the continuum between the highly-regularized initial fit and the low-bias solution, adapting to the specific signal-to-noise ratio and covariance structure of the data, offering a more general, data-driven solution to the double penalization problem that concerned \cite{Zou2005}.

To explicitly quantify Renet's combined debiasing effect, we derive the following recovery ratio.

\begin{corollary}[Signal Recovery Ratio] \label{cor:Recovery}
Let $R(\theta) = \frac{|\hat{\beta}_{Renet, j}(\theta)|}{|\hat{\beta}_{OLS, j}|}$ denote the fraction of the unconstrained signal magnitude preserved by the Renet estimator for a selected feature where $|\hat{\beta}_{OLS, j}| > \theta \lambda \alpha$. Under the orthogonal design assumptions of Proposition \ref{prop:Bias}, this ratio is given by:
\begin{equation} \label{eq:RecoveryRatio}
    R(\theta) = \frac{1 - \frac{\theta \lambda \alpha}{|\hat{\beta}_{OLS, j}|}}{1 + \theta \lambda (1-\alpha)}
\end{equation}
Furthermore, for any fixed $\lambda > 0$, $\alpha \in [0, 1)$, and $|\hat{\beta}_{OLS, j}| > 0$, the recovery ratio $R(\theta)$ is a monotonically decreasing function of $\theta$ on the interval $(0, 1]$, with $\lim_{\theta \to 0^+} R(\theta) = 1$.
\end{corollary}

\begin{proof}
The expression for $R(\theta)$ follows directly from dividing the Renet solution in Proposition \ref{prop:Bias} by $|\hat{\beta}_{OLS, j}|$. To show monotonicity, we take the partial derivative with respect to $\theta$:
\begin{equation}
    \frac{\partial R}{\partial \theta} = - \frac{\lambda \left(\frac{\alpha}{|\hat{\beta}_{OLS, j}|}  + (1-\alpha) \right)}{[1 + \theta \lambda (1-\alpha)]^2}
\end{equation}
Since $\lambda, \alpha, (1-\alpha),$ and $|\hat{\beta}_{OLS, j}|$ are all positive, $\frac{\partial R}{\partial \theta} < 0$ for all $\theta \in (0, 1]$. The limit as $\theta \to 0^+$ clearly yields $R(\theta) = 1$, representing the recovery of the unbiased OLS estimate.
\end{proof}

\subsection{The Synergy Between Relaxation and the 1-SE Rule}
\label{sec:1SE}
A key theoretical insight of the present work is the relationship between the ``One-Standard-Error'' (1-SE) rule and relaxation.

The 1-SE rule is a well-known heuristic used to select a sparser model than the one that minimizes out-of-fold cross-validation error ($CV_{min}$). It selects the largest $\lambda$ (strongest regularization) such that the cross-validation error is within one standard error of the minimum. This typically results in:
\begin{itemize}
    \item \textbf{Higher Precision (Sparsity):} False positives are reduced as the active set shrinks.
    \item \textbf{Higher Bias (Shrinkage):} The surviving coefficients are heavily penalized, often leading to an increase in mean squared prediction error compared to the dense $CV_{min}$ model.
\end{itemize}

Renet resolves this trade-off. By introducing $\theta$, we decouple variable selection from coefficient estimation. Let $\lambda_{1SE}$ determine the active set $\mathcal{A}_{1SE}$. While the 1-SE rule tends to better recover the true support (by ignoring statistically insignificant accuracy gains in favor of a sparser representation) it severely underestimates coefficient magnitudes. In Renet, the relaxation parameter $\theta$ effectively scales the penalty down to $\theta \lambda_{1SE}$ \textit{without} changing the active set.

This mechanism allows Renet to ``stretch the budget'' of the 1-SE rule: for any given level of coefficient shrinkage (bias) we can achieve a sparser model than if the 1-SE rule was used with standard Elastic Net; conversely, given a sparsity level induced by the 1-SE rule,  Renet coefficients can enjoy lower bias. Indeed, in Renet we can select a very parsimonious model by choosing a large $\lambda$ and then repair the shrinkage damage by choosing a small $\theta$, achieving a point on the bias-variance curve that is inaccessible to the standard Elastic Net, regardless of the choice of $\lambda$.  Consequently, the 1-SE rule (often) provides an improved active set for relaxation, while relaxation eliminates the primary predictive penalty of the 1-SE rule (excessive regularization) -- a mutual reinforcement that justifies the ``synergy'' alluded to at the start of this section.

We formalize this advantage in the following proposition, which shows that if the 1-SE rule successfully screens out noise, Renet with 1-SE rule dominates the standard Elastic Net in expectation.

\begin{proposition}[1-SE Renet Dominance Hierarchy] \label{prop:Hierarchy}
Let $\mathcal{S}$ be the true support of the data-generating process. Assume the 1-SE rule correctly identifies the true support ($\mathcal{A}_{1SE} = \mathcal{S}$), whereas the standard $CV_{min}$ rule selects a superset $\mathcal{A}_{min} \supseteq \mathcal{S}$. Then, there exists a $\theta^* \in (0, 1]$ such that, for any CV-optimized $\theta' \in (0, 1]$ the following hierarchy holds in expectation:
\begin{equation}
    MSE(\hat{\beta}_{Renet}(\lambda_{1SE}, \theta^*)) \le MSE(\hat{\beta}_{Renet}(\lambda_{min}, \theta')) \le MSE(\hat{\beta}_{EN}(\lambda_{min})) \le MSE(\hat{\beta}_{EN}(\lambda_{1SE}))
\end{equation}
The first inequality is strict whenever $\mathcal{A}_{min} \supset \mathcal{S}$.
\end{proposition}

\begin{proof}
The right-most inequality is true by the definition of $\lambda_{min}$ as the global minimizer of the cross-validation error surface for the standard Elastic Net. The middle inequality  holds because the Renet search space over $\theta \in (0, 1]$ includes the standard Elastic Net solution ($\theta=1$); thus, an optimized $\theta'$ must yield a risk no greater than the standard $CV_{min}$ fit.

For the left-most (and crucial) inequality, the case of equality is trivial, so we restrict attention to the $\mathcal{A}_{min} \supset \mathcal{S}$ case.  Consider the risk decomposition $R(\hat{\beta}) = B^2(\hat{\beta}) + Var(\hat{\beta})$. Under the assumption $\mathcal{A}_{1SE} = \mathcal{S}$, the Renet estimator $\hat{\beta}_{Renet}(\lambda_{1SE}, \theta)$ can approach the Oracle risk $R(\hat{\beta}_{Oracle}) = \sigma^2 \text{tr}((\mathbf{X}_{\mathcal{S}}^T \mathbf{X}_{\mathcal{S}})^{-1})$ as $\theta \to 0$. In contrast, any estimator operating on the noisy superset $\mathcal{A}_{min} \supset \mathcal{S}$ is subject to variance inflation $Var(\mathcal{A}_{min}) > Var(\mathcal{S})$. Specifically, the variance of an unbiased (relaxed) estimator increases with the inclusion of irrelevant predictors. Even with an optimized $\theta'$, the estimator on $\mathcal{A}_{min}$ must either accept this higher variance or increase $\theta'$ to suppress it, which necessarily re-introduces bias. Since the 1-SE support is the minimal sufficient set, its optimal relaxation $\theta^*$ achieves a lower bound on risk that is structurally inaccessible to any model containing noise variables.
\end{proof}

This hierarchy reflects that the perceived weakness of the 1-SE rule---its heavy shrinkage---vanishes in Renet. Once decoupled via relaxation, the 1-SE rule's (often) superior precision becomes a predictive advantage, allowing Renet to achieve a risk profile that is inaccessible to denser, more regularized models like the standard Elastic Net. 

Lastly, it should be noted that while we have focused on the standard 1-SE rule, the multiplier need not be unity: the degree of conservative selection can be adjusted (e.g., a $k$-SE rule) to reflect specific precision-recall requirements or accuracy-sparsity goals. Regardless of the multiplier chosen, relaxation will minimize the induced loss in predictive accuracy, if there is any.

\section{Empirical Results}
\label{sec:Empirics}

\subsection{Formal Setup}
We consider a varied selection of 10 real-world benchmark datasets (Table \ref{tab:datasets_real}) and 10 synthetic datasets (Table \ref{tab:datasets_synth}) that include a wide range of scenarios covering small-to-large sample sizes, dimensionality, supports and signal-to-noise ratios (SNR).  Synthetic scenario $S_1$ serves as baseline ($n \gg p$, high SNR), where any reasonable linear model should perform well. On the opposite end of the spectrum, scenarios $S_7$ through $S_{10}$ specifically target the $p \geq n$ regime, which is more challenging for linear models and is designed to test the limits of OLS relaxation, where the stability of the estimator is most vulnerable.

As customary in the literature (e.g. \cite{Zou2005, Meinshausen2007, Zou2009}), for all synthetic scenarios $S_1$ through $S_{10}$, we generate the response variable following a linear model $y = X\beta + \varepsilon$, where $\varepsilon \sim \mathcal{N}(0,\sigma^2 \mathbf{I})$, $\sigma$ being scenario-dependent. Features are sampled from a multivariate normal distribution $X \sim \mathcal{N}(0,\Sigma)$, with the covariance structure $\Sigma$ specified in Table \ref{tab:datasets_synth} for each scenario. Concretely,  $\Sigma_\rho$ denotes an identity matrix with $\rho$ in its off-diagonal entries. On the other hand, $\Sigma_{\rho, T}$ is an identity matrix with decaying off-diagonal values (Toeplitz structure), starting at $\rho$. Lastly, $\Sigma_{\rho, B}$ denotes an identity matrix with disjoint correlation blocks for signal and noise features, having pairwise intra-block correlations of $\rho$ and $\rho/2$, respectively.

To maintain the original dimensionality, categorical variables are target-encoded following a cross fitting scheme to prevent target leakage and overfitting. Empirical Bayes smoothing is applied to shrink category-level estimates towards the global mean, mitigating the high variance associated with low-cardinality levels.

In order to ensure our results are not sensitive to a specific random state, we employ three seeds: 123 and 321 (for consistency with our previous work) and 42 (a standard reference in the literature). The reported standard errors are pooled to incorporate both within-dataset variability (folds) and between-seed variability.

\begin{table}[ht]
\centering
\caption{Real-world datasets. $n$ denotes the sample size, $p$ the total number of features, and $c$ the number of categorical predictors out of $p$.}
\label{tab:datasets_real}
\begin{tabular}{@{}l l l l l@{}}
\toprule
\textbf{Name} & \textbf{n} & \textbf{p} & \textbf{c} & \textbf{Source} \\ 
\midrule
Abalone & 4177 & 8 & 1 & \href{https://archive.ics.uci.edu/dataset/1/abalone}{UCI Machine Learning Repository} \\
BMI & 136 & 11 & 1 & Applied Linear Regression \citep{Weisberg2014} \\
College & 775 & 17 & 1 & ISLR \citep{James2013} \\
Crime & 1994 & 100 & 1 & \href{https://archive.ics.uci.edu/dataset/183/communities+and+crime}{UCI Machine Learning Repository} \\
Diabetes & 442 & 10 & 1 & Least Angle Regression \citep{Efron2004} \\
Doctor & 5190 & 14 & 3 & Econometric models \citep{Cameron1986} \\
Parkinson & 5875 & 19 & 1 & \href{https://archive.ics.uci.edu/dataset/189/parkinsons+telemonitoring}{UCI Machine Learning Repository} \\
Riboflavin & 71 & 4088 & 0 & High-dimensional statistics \citep{Bühlmann2014} \\
Student & 649 & 32 & 17 & \href{https://archive.ics.uci.edu/dataset/320/student+performance}{UCI Machine Learning Repository} \\
Wine & 6497 & 11 & 0 & \href{https://archive.ics.uci.edu/dataset/186/wine+quality}{UCI Machine Learning Repository} \\
\bottomrule
\end{tabular}
\end{table}

\begin{table}[ht]
\centering
\caption{Synthetic datasets. $n$ denotes the sample size, $p$ the number of features,  $s$ the number of informative (signal) features,  $\rho$ the correlation among features ($T=$ Toeplitz, $B=$ Block-correlated),  and $\sigma$ the standard deviation of the Gaussian noise.}
\label{tab:datasets_synth}
\begin{tabular}{@{}l l l l l l l@{}}
\toprule
\textbf{Name} & \textbf{n} & \textbf{p} & \textbf{s} & $\mathbf{\rho}$ & $\mathbf{\sigma}$ & $\mathbf{\Sigma}$ \\ 
\midrule
S1 & 100,000 & 20 & 10 & 0 & 1 & $\mathbf{I}$\\
S2 & 5,000 & 20 & 5 & 0.5 & 2 & $\Sigma_\rho$\\
S3 & 2,000 & 20 & 2 & 0.75 & 0.5 & $\Sigma_\rho$\\
S4 & 1,000 & 100 & 10 & 0.75 ($T$) & 2 & $\Sigma_{\rho, T}$\\
S5 & 5,000 & 100 & 80 & 0.5 & 2 & $\Sigma_\rho$\\
S6 & 500 & 20 & 2 & 0.25 & 1 & $\Sigma_\rho$\\
S7 & 300 & 300 & 10 & 0.5 & 1 & $\Sigma_\rho$\\
S8 & 90 & 4,000 & 100 & 0.25 & 0.25 & $\Sigma_\rho$\\
S9 & 200 & 220 & 20 & 0.75 & 2 & $\Sigma_\rho$\\
S10 & 300 & 3,000 & 30 & 0.75 ($B$) & 2 & $\Sigma_{\rho, B}$\\
\bottomrule
\end{tabular}
\end{table}

Our benchmarks include algorithms that are meaningfully related to Renet. The most relevant comparisons are kept in the main body of this article, while the rest are included in Appendix \ref{ap:ExRes}. The algorithms that fall in the first category are industry-standard and listed below alongside our proposed Renet algorithm. 
\begin{itemize}
\item \texttt{scikit-learn}'s \texttt{ElasticNetCV} with and without 1-SE rule
\item Adaptive Elastic Net (AEN) with a cross-validated Ridge pilot
\item Renet with and without 1-SE rule
\end{itemize}

Our Python implementation of the Adaptive Elastic Net (AEN) utilizes a feature-weighting transformation rather than a custom coordinate descent solver. Unlike the original formulation by \cite{Zou2005} -- which typically applies adaptive weights to the $\ell_1$ component only -- our approach effectively scales both the $\ell_1$ and $\ell_2$ penalties for each coefficient by a common feature-specific weight $w_j$. In practice, the deviation from the original AEN is expected to be negligible: since our $\ell_1$-ratio is set to $0.95$, the $\ell_1$ penalty is weighted 19 times more heavily than the $\ell_2$ penalty. The latter is primarily retained to provide numerical stability and ensure a unique solution path in the $p > n$ regime. Pseudocode for this procedure is provided in Appendix \ref{ap:Algo_AEN}, and implemented  in the \texttt{AdaptiveNetCV} class of the \texttt{trust-free} package (version $\geq$ 3.0.0).

\bigskip

\noindent The benchmarks presented in Appendix \ref{ap:ExRes} compare Renet against additional baselines:
\begin{itemize}
    \item \textbf{\texttt{Celer} Elastic Net:} The \texttt{ElasticNetCV} implementation from the specialized \texttt{celer} library,  evaluated both at the cross-validation minimum and via the 1-SE rule. 
    \item \textbf{Stepwise OLS:} A forward-backward sequential feature selector using \texttt{mlxtend} \citep{Raschka2018} wrapped around an OLS objective. This serves as a classical,  non-regularized baseline for sparsity and computational scaling.
    \item \textbf{Renet (ablated):} To isolate the performance gains of our relaxation logic, we include \textit{Renet without relaxation}. This configuration utilizes Renet's adaptive solver backend,  which dynamically dispatches the optimization task to either \texttt{scikit-learn} or \texttt{celer} depending on the topology of the data. 
\end{itemize}

\noindent To ensure full reproducibility, the real-world datasets, the synthetic data generation suite, and the scripts used for benchmarking are provided in the accompanying repository at \href{https://github.com/adc-trust-ai/trust-free/}{\texttt{github.com/adc-trust-ai/trust-free}}.

\subsection{Results and Discussion}

We derive the following key conclusions from the benchmarking results presented in Tables \ref{tab:main_results_Real} and \ref{tab:main_results_Synth}:

\begin{itemize}
    \item \textbf{Accuracy, sparsity, and the relaxation effect:} In high-SNR regimes (e.g., \textit{BMI}, \textit{College}, \textit{S1}, \textit{S5}), both AEN and Renet consistently outperform standard Elastic Net, which suffers from inherent $\ell_1$ and $\ell_2$ estimation bias. Notably, in more complex, high-dimensional scenarios such as \textit{Riboflavin} ($n=71$, $p=4,088$) and \textit{S10}, Renet demonstrates superior predictive robustness. By allowing the relaxation parameter $\theta$ to tend toward $0$, Renet recovers the signal amplitude that standard Elastic Net mutes through over-regularization. The gains in accuracy and sparsity of Renet relative to the ablated Renet baseline in Appendix \ref{ap:ExRes} are, similarly, very significant. When combined with the 1-SE rule, Renet delivers a ``complexity collapse'' that standard methods cannot replicate. In the \textit{Crime} and \textit{Student} datasets, Renet with the 1-SE rule achieves comparable $R^2$ to the CV-min Elastic Net while using approximately $85\%$ fewer features. While AEN remains an excellent alternative in high-SNR environments, it proves fragile in ultra-high-dimensional problems (e.g.,  \textit{Riboflavin}, \textit{S8}) or lower-SNR and highly correlated settings (e.g., \textit{S7}, \textit{S9}, \textit{S10}), where Renet's performance remains stable. 
        
    \item \textbf{Validation of the 1-SE Logic:} In all tested cases, Renet with the 1-SE rule yields models with comparable or higher accuracy than the standard Elastic Net with the 1-SE rule while simultaneously being sparser, often by a wide margin. This provides strong empirical evidence for our theoretical argument in Section \ref{sec:1SE}: the 1-SE rule is significantly more effective when the candidate models are permitted to debias their coefficients through relaxation.
    
    \item \textbf{Computational Efficiency:} Despite the sophisticated relaxation logic, Renet's speed remains competitive. This efficiency stems from two primary sources. First, the overhead of the occasional Elastic Net refits is negligible; the sub-problems involve only the sparse support identified in the initial path and converge in very few iterations due to warm-started coefficients and a reduced array of candidate penalty values. Second, Renet dynamically switches between the \texttt{scikit-learn} and \texttt{celer} solvers based on the data topology and the $n/p$ ratio. Concretely, in high-dimensional cases like \textit{Riboflavin}, Renet can be significantly faster than the standard CV-min path ($0.471$s vs $1.468$s), as the $n \ll p$ regime triggers the switch to the \texttt{celer} solver, which is state-of-the-art in this domain.
\end{itemize}

Overall, as Figure \ref{fig:renet_value_frontier} illustrates as well,  Renet consistently improves upon the standard Elastic Net across both accuracy and sparsity. When the 1-SE rule is applied, Renet exploits a unique synergy between selection and relaxation to deliver extreme sparsity -- exemplified by the \textit{Student} dataset where a single feature captures $84\%$ of the variance without compromising predictive power (less than $0.5\%$ below the highest scoring algorithm). While AEN remains a very competitive alternative in well-behaved regimes, Renet's adaptive relaxation makes it a robust performer across all scenarios tested, including low SNR, high multicollinearity, and ultra high-dimensional problems.

\begin{figure}[h!]
    \centering
    \begin{subfigure}[b]{0.493\textwidth}
        \centering
        \includegraphics[width=\textwidth]{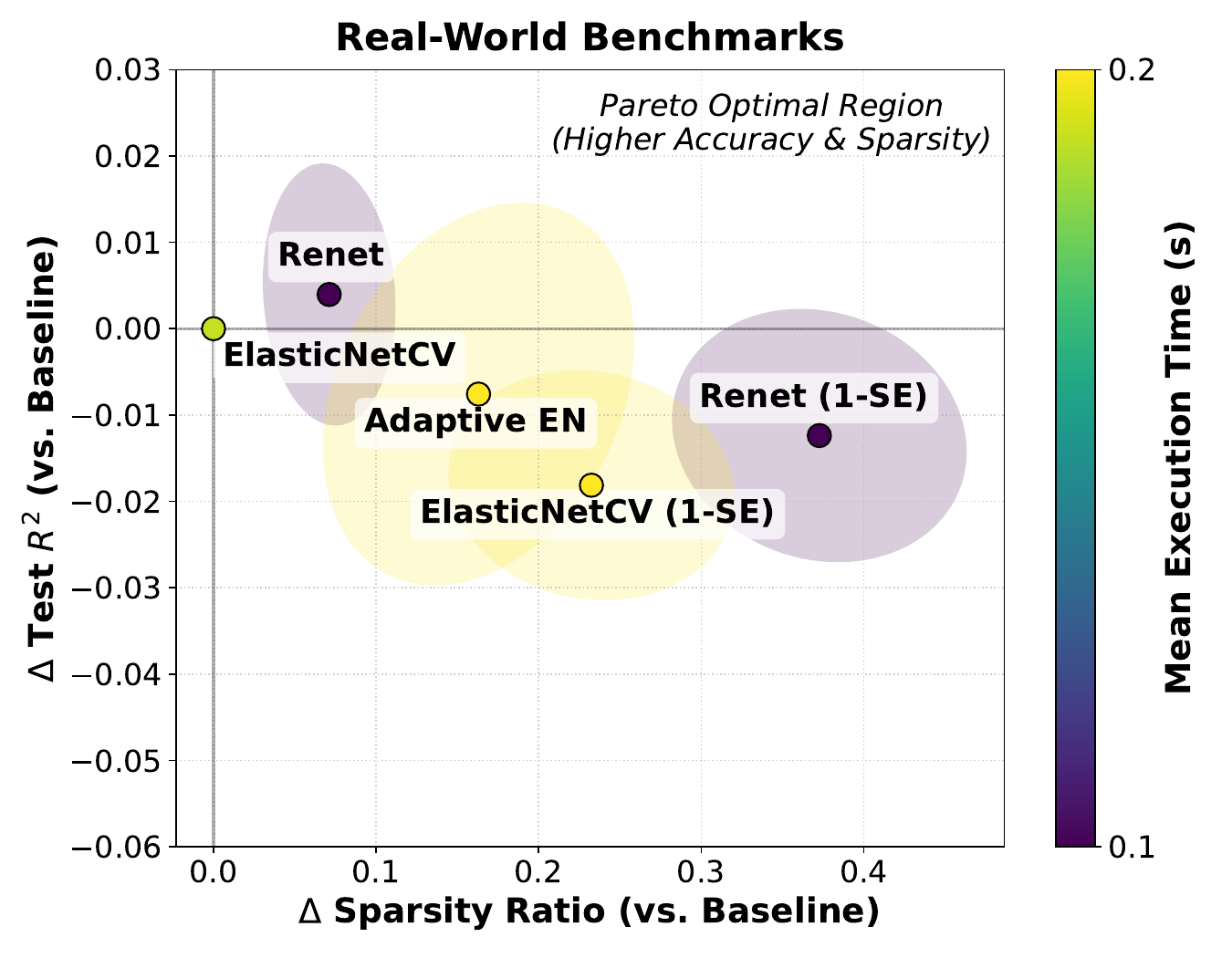}
        \caption{Real-World Benchmarks}
        \label{fig:delta_real}
    \end{subfigure}
    \hfill
    \begin{subfigure}[b]{0.493\textwidth}
        \centering
        \includegraphics[width=\textwidth]{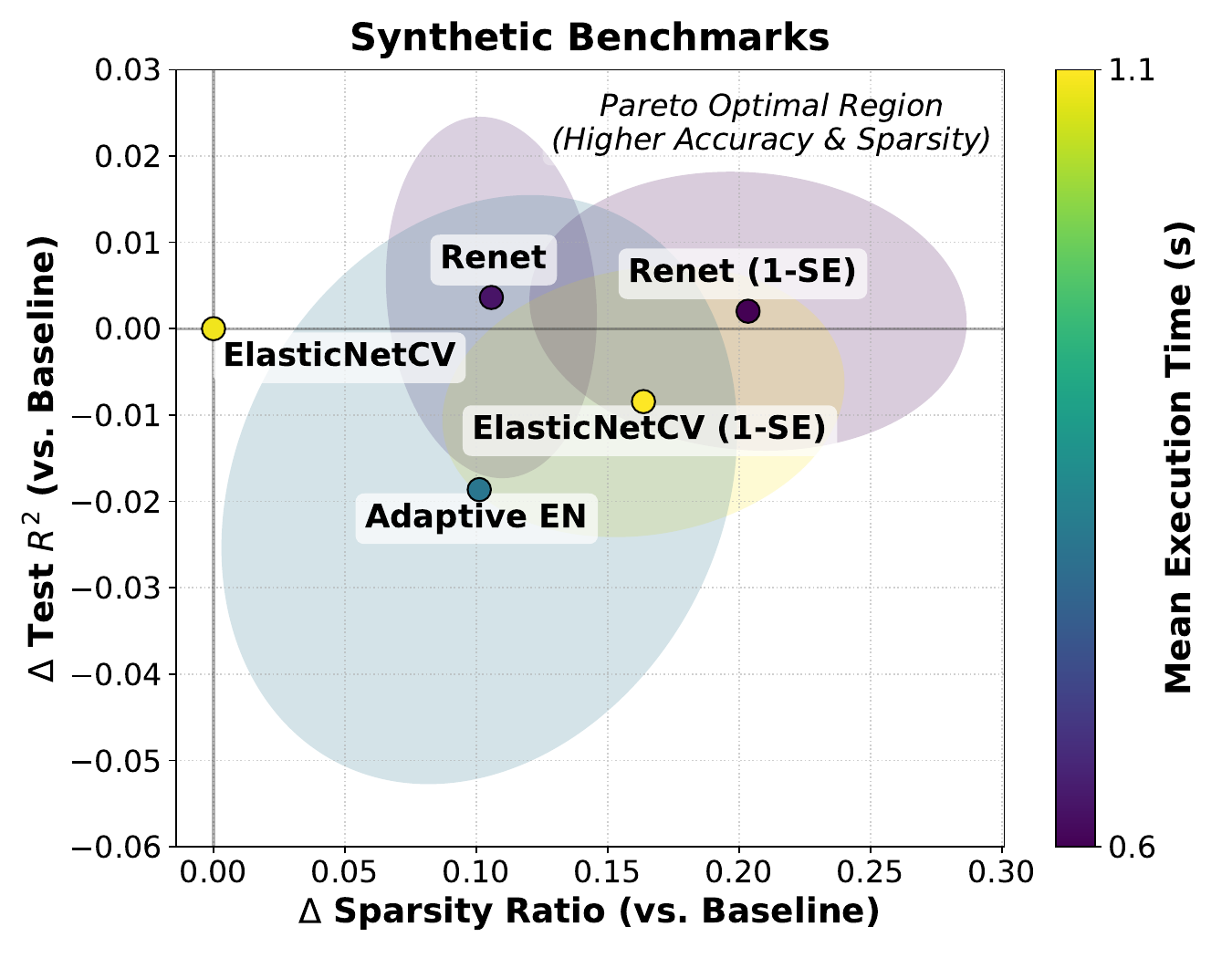}
        \caption{Synthetic Benchmarks}
        \label{fig:delta_synthetic}
    \end{subfigure}
    
    \caption{Performance is visualized as the absolute change relative to \textit{scikit-learn}'s Elastic Net (centered at 0,0). Ellipses show a 1-SE  confidence region of the bivariate mean, pooled across datasets, folds, and seeds. }
    \label{fig:renet_value_frontier}
\end{figure}

\setlength{\tabcolsep}{4pt} 
\begin{table}[h!]
\centering
\caption{Real-world 10-fold out-of-sample results with standard errors in parenthesis. Best metrics in bold.}
\label{tab:main_results_Real}
\small 
\begin{tabular}{@{}l l c c c c c@{}}
\toprule
\textbf{Dataset} & \textbf{Metric} & \textbf{AEN} & \textbf{Scikit EN} & \textbf{Scikit EN, 1-SE} & \textbf{Renet} & \textbf{Renet, 1-SE} \\ 
\midrule
\multirow{4}{*}{Abalone} & $R^2$ & 0.523 (0.010) & \textbf{0.525 (0.010)} & 0.510 (0.010) & \textbf{0.525 (0.010)} & 0.514 (0.010) \\
 & Num. Coeff.  & 7.0 (0.0) & 7.3 (0.1) & 6.2 (0.2) & 7.1 (0.1) & \textbf{5.5 (0.2)} \\
 & Relaxation ($\theta$)  & 1.0  & 1.0  & 1.0  & 0.1 (0.0) & 0.0 (0.0) \\
 & Time (s)  & 0.024 & 0.019 & 0.020 & \textbf{0.014} & \textbf{0.014} \\ \midrule

\multirow{4}{*}{BMI} & $R^2$ & \textbf{0.990 (0.001)} & 0.989 (0.001) & 0.988 (0.001) & 0.989 (0.001) & \textbf{0.990 (0.001)} \\
 & Num. Coeff.  & \textbf{2.0 (0.0)} & 6.6 (0.3) & 4.6 (0.2) & 4.5 (0.5) & 3.0 (0.0) \\
 & Relaxation ($\theta$)  & 1.0 & 1.0  & 1.0  & 0.2 (0.1) & 0.0 (0.0) \\
 & Time (s)  & 0.009 & \textbf{0.004} & \textbf{0.004} & 0.008 & 0.008 \\ \midrule

\multirow{4}{*}{College} & $R^2$ & 0.909 (0.008) & 0.893 (0.009) & 0.858 (0.009) & \textbf{0.911 (0.008)} & 0.901 (0.010) \\
 & Num. Coeff.  & \textbf{6.5 (0.3)} & 16.7 (0.1) & 16.9 (0.1) & 16.4 (0.3) & 8.7 (0.3) \\
 & Relaxation ($\theta$)  & 1.0  & 1.0  & 1.0  & 0.0 (0.0) & 0.0 (0.0) \\
 & Time (s)  & 0.011 & \textbf{0.005} & 0.006 & 0.034 & 0.033 \\ \midrule

\multirow{4}{*}{Crime} & $R^2$ & 0.663 (0.010) & \textbf{0.664 (0.009)} & 0.646 (0.010) & 0.663 (0.009) & 0.645 (0.010) \\
 & Num. Coeff.  & 54.7 (1.4) & 65.7 (2.8) & 13.9 (0.6) & 63.2 (3.7) & \textbf{8.7 (0.7)} \\
 & Relaxation ($\theta$)  & 1.0  & 1.0  & 1.0  & 0.8 (0.0) & 0.0 (0.0) \\
 & Time (s)  & \textbf{0.057} & 0.058 & 0.076 & 0.539 & 0.546 \\ \midrule

\multirow{4}{*}{Diabetes} & $R^2$ & 0.472 (0.021) & \textbf{0.473 (0.021)} & 0.446 (0.018) & 0.471 (0.021) & 0.443 (0.024) \\
 & Num. Coeff.  & 7.0 (0.2) & 8.8 (0.2) & 6.8 (0.2) & 8.1 (0.3) & \textbf{4.5 (0.6)} \\
 & Relaxation ($\theta$)  & 1.0  & 1.0  & 1.0  & 0.3 (0.1) & 0.0 (0.0) \\
 & Time (s)  & 0.010 & 0.006 & \textbf{0.005} & 0.014 & 0.014 \\ \midrule

\multirow{4}{*}{Doctor} & $R^2$ & \textbf{0.220 (0.013)} & \textbf{0.220 (0.012)} & 0.184 (0.007) & \textbf{0.220 (0.013)} & 0.186 (0.012) \\
 & Num. Coeff.  & 7.1 (0.4) & 9.4 (0.8) & 2.8 (0.1) & 7.2 (1.1) & \textbf{1.6 (0.1)} \\
 & Relaxation ($\theta$)  & 1.0  & 1.0  & 1.0  & 0.4 (0.1) & 0.0 (0.0) \\
 & Time (s)  & 0.022 & \textbf{0.008} & \textbf{0.008} & 0.016 & 0.014 \\ \midrule

\multirow{4}{*}{Parkinson} & $R^2$ & \textbf{0.168 (0.005)} & \textbf{0.168 (0.005)} & 0.159 (0.005) & \textbf{0.168 (0.005)} & 0.159 (0.006) \\
 & Num. Coeff.  & 17.1 (0.3) & 16.5 (0.3) & 12.3 (0.6) & 16.0 (0.3) & \textbf{10.3 (0.6)} \\
 & Relaxation ($\theta$)  & 1.0  & 1.0  & 1.0  & 0.5 (0.1) & 0.0 (0.0) \\
 & Time (s)  & 0.075 & \textbf{0.062} & 0.063 & 0.109 & 0.125 \\ \midrule

\multirow{4}{*}{Riboflavin} & $R^2$ & 0.520 (0.064) & 0.607 (0.046) & 0.608 (0.042) & \textbf{0.630 (0.045)} & 0.594 (0.051) \\
 & Num. Coeff.  & 123.6 (2.5) & 46.9 (2.8) & 25.2 (1.2) & 34.7 (2.9) & \textbf{17.5 (2.0)} \\
 & Relaxation ($\theta$)  & 1.0  & 1.0  & 1.0  & 0.6 (0.0) & 0.5 (0.0) \\
 & Time (s)  & 1.448 & 1.468 & 1.488 & 0.487 & \textbf{0.471} \\ \midrule

\multirow{4}{*}{Student} & $R^2$ & 0.843 (0.014) & 0.843 (0.014) & 0.816 (0.014) & \textbf{0.845 (0.013)} & 0.841 (0.014) \\
 & Num. Coeff.  & 4.2 (0.7) & 8.2 (1.5) & 2.0 (0.0) & 2.0 (0.0) & \textbf{1.0 (0.0)} \\
 & Relaxation ($\theta$)  & 1.0  & 1.0  & 1.0 & 0.0 (0.0) & 0.0 (0.0) \\
 & Time (s)  & 0.012 & \textbf{0.006} & \textbf{0.006} & 0.012 & 0.012 \\ \midrule

\multirow{4}{*}{Wine} & $R^2$ & \textbf{0.287 (0.005)} & \textbf{0.287 (0.005)} & 0.273 (0.005) & \textbf{0.287 (0.005)} & 0.274 (0.005) \\
 & Num. Coeff.  & 10.4 (0.1) & 11.0 (0.0) & 6.4 (0.5) & 10.9 (0.1) & \textbf{4.5 (0.4)} \\
 & Relaxation ($\theta$)  & 1.0  & 1.0  & 1.0  & 0.1 (0.0) & 0.1 (0.0) \\
 & Time (s)  & 0.021 & \textbf{0.010} & 0.012 & 0.018 & 0.014 \\
\bottomrule
\end{tabular}
\end{table}

\setlength{\tabcolsep}{4pt} 
\begin{table}[h!]
\centering
\caption{Synthetic 10-fold out-of-sample results with standard errors in parenthesis. Best metrics in bold.}
\label{tab:main_results_Synth}
\small 
\begin{tabular}{@{}l l c c c c c@{}}
\toprule
\textbf{Dataset} & \textbf{Metric} & \textbf{AEN} & \textbf{Scikit EN} & \textbf{Scikit EN, 1-SE} & \textbf{Renet} & \textbf{Renet, 1-SE} \\ 
\midrule
\multirow{4}{*}{S1} & $R^2$ & 0.932 (0.000) & \textbf{0.933 (0.000)} & 0.932 (0.000) & \textbf{0.933 (0.000)} & \textbf{0.933 (0.000)} \\
 & Num. Coeff.  & \textbf{10.0 (0.0)} & 15.1 (0.3) & \textbf{10.0 (0.0)} & 10.2 (0.1) & \textbf{10.0 (0.0)} \\
 & Relaxation ($\theta$)  & 1.0 & 1.0  & 1.0  & 0.0 (0.0) & 0.0 (0.0) \\
 & Time (s)  & 0.203 & \textbf{0.085} & 0.099 & 0.173 & 0.179 \\ \midrule

\multirow{4}{*}{S2} & $R^2$ & \textbf{0.862 (0.002)} & \textbf{0.862 (0.002)} & 0.860 (0.002) & \textbf{0.862 (0.002)} & \textbf{0.862 (0.002)} \\
 & Num. Coeff.  & 5.6 (0.1) & 17.4 (0.4) & 6.3 (0.3) & 13.9 (1.2) & \textbf{5.0 (0.0)} \\
 & Relaxation ($\theta$)  & 1.0  & 1.0  & 1.0  & 0.5 (0.1) & 0.0 (0.0) \\
 & Time (s)  & 0.023 & 0.013 & \textbf{0.011} & 0.020 & 0.020 \\ \midrule

\multirow{4}{*}{S3} & $R^2$ & \textbf{0.396 (0.011)} & 0.395 (0.011) & 0.381 (0.008) & 0.393 (0.011) & 0.394 (0.010) \\
 & Num. Coeff.  & 5.4 (0.2) & 5.0 (0.6) & 1.7 (0.2) & 3.6 (0.4) & \textbf{1.0 (0.0)} \\
 & Relaxation ($\theta$)  & 1.0  & 1.0  & 1.0  & 0.3 (0.1) & 0.0 (0.0) \\
 & Time (s)  & 0.017 & \textbf{0.007} & \textbf{0.007} & 0.013 & 0.015 \\ \midrule

\multirow{4}{*}{S4} & $R^2$ & \textbf{0.886 (0.003)} & 0.884 (0.003) & 0.881 (0.003) & 0.884 (0.003) & 0.883 (0.003) \\
 & Num. Coeff.  & 15.2 (0.6) & 24.5 (0.7) & 10.2 (0.5) & 13.1 (0.9) & \textbf{8.5 (0.3)} \\
 & Relaxation ($\theta$)  & 1.0  & 1.0  & 1.0  & 0.3 (0.0) & 0.2 (0.0) \\
 & Time (s)  & 0.023 & 0.012 & \textbf{0.011} & 0.052 & 0.051 \\ \midrule

\multirow{4}{*}{S5} & $R^2$ & 0.893 (0.002) & \textbf{0.895 (0.002)} & 0.893 (0.002) & \textbf{0.895 (0.002)} & 0.893 (0.002) \\
 & Num. Coeff.  & 75.9 (0.2) & 96.1 (0.5) & 84.2 (0.4) & 80.9 (0.6) & \textbf{71.3 (0.5)} \\
 & Relaxation ($\theta$)  & 1.0 & 1.0  & 1.0  & 0.0 (0.0) & 0.0 (0.0) \\
 & Time (s)  & 0.347 & 0.404 & 0.582 & 0.265 & \textbf{0.240} \\ \midrule

\multirow{4}{*}{S6} & $R^2$ & \textbf{0.760 (0.015)} & \textbf{0.760 (0.015)} & 0.752 (0.014) & 0.758 (0.015) & \textbf{0.760 (0.015)} \\
 & Num. Coeff.  & 5.2 (0.4) & 6.4 (0.4) & 2.9 (0.1) & 3.5 (0.4) & \textbf{2.0 (0.0)} \\
 & Relaxation ($\theta$)  & 1.0  & 1.0  & 1.0  & 0.1 (0.0) & 0.0 (0.0) \\
 & Time (s)  & 0.010 & \textbf{0.005} & \textbf{0.005} & 0.009 & 0.008 \\ \midrule

\multirow{4}{*}{S7} & $R^2$ & 0.849 (0.009) & 0.859 (0.006) & 0.852 (0.006) & 0.867 (0.007) & \textbf{0.869 (0.007)} \\
 & Num. Coeff.  & 54.6 (5.4) & 37.2 (1.5) & 22.9 (0.8) & 11.1 (1.4) & \textbf{8.7 (0.2)} \\
 & Relaxation ($\theta$)  & 1.0  & 1.0  & 1.0  & 0.0 (0.0) & 0.0 (0.0) \\
 & Time (s)  & \textbf{0.965} & 3.627 & 3.638 & 2.645 & 2.613 \\ \midrule

\multirow{4}{*}{S8} & $R^2$ & -0.019 (0.106) & 0.134 (0.065) & 0.105 (0.068) & \textbf{0.160 (0.077)} & 0.152 (0.067) \\
 & Num. Coeff.  & 167.4 (1.6) & 38.7 (3.3) & 16.5 (1.8) & 20.6 (2.0) & \textbf{8.5 (0.9)} \\
 & Relaxation ($\theta$)  & 1.0  & 1.0  & 1.0  & 0.5 (0.0) & 0.5 (0.0) \\
 & Time (s)  & 1.217 & 0.760 & 0.757 & 0.331 & \textbf{0.327} \\ \midrule

\multirow{4}{*}{S9} & $R^2$ & 0.900 (0.013) & 0.902 (0.014) & 0.896 (0.016) & \textbf{0.903 (0.015)} & 0.900 (0.014) \\
 & Num. Coeff.  & 70.6 (2.1) & 70.3 (6.8) & 34.4 (5.0) & 57.4 (8.1) & \textbf{21.6 (3.2)} \\
 & Relaxation ($\theta$)  & 1.0 & 1.0  & 1.0  & 0.7 (0.1) & 0.2 (0.1) \\
 & Time (s)  & \textbf{0.654} & 1.957 & 1.939 & 1.101 & 1.041 \\ \midrule

\multirow{4}{*}{S10} & $R^2$ & 0.734 (0.013) & 0.756 (0.012) & 0.744 (0.010) & \textbf{0.762 (0.013)} & 0.754 (0.014) \\
 & Num. Coeff.  & 261.9 (1.6) & 8.5 (0.6) & 5.8 (0.2) & 5.5 (0.1) & \textbf{3.6 (0.4)} \\
 & Relaxation ($\theta$)  & 1.0  & 1.0  & 1.0  & 0.0 (0.0) & 0.0 (0.0) \\
 & Time (s)  & 3.807 & 3.781 & 3.760 & 1.220 & \textbf{1.136} \\
\bottomrule
\end{tabular}
\end{table}

On the other hand, the extended results in Tables~\ref{tab:appendix_results_Real} and \ref{tab:appendix_results_Synth} (Appendix~\ref{ap:ExRes}) demonstrate that Stepwise OLS poses significant scalability challenges; it is consistently at least one to two orders of magnitude slower than the regularized benchmarks and frequently fails to converge within a 30-minute time limit in high-dimensional settings while the rest of algorithms finish in seconds (or less). Furthermore, these results serve to validate the base Renet engine: Renet without relaxation yields results numerically identical to state-of-the-art coordinate descent implementations like \texttt{Celer} in terms of $R^2$ and sparsity, while maintaining a significant speed advantage in low-to-moderate dimensional regimes. This is achieved through Renet's hybrid backend, which dynamically leverages \texttt{scikit-learn} for efficiency in standard regimes and switches to \texttt{Celer} as $p$ grows, effectively providing the ``best of both worlds'' in computational performance.

\clearpage

\section{Conclusion}
\label{sec:conclusion}

In the present work we have presented Renet, a principled,  stability-enhanced implementation of the Relaxed Elastic Net that systematically addresses the bias-variance conflict inherent in penalized linear regression. By rigorously decoupling variable selection from coefficient estimation, Renet allows the researcher to enjoy the benefits of strong regularization during the selection phase -- screening out noise and handling multicollinearity -- without suffering the predictive penalty of shrinkage bias in the final model. 

Our theoretical and empirical results highlight three distinct advantages of this framework.

First, a unique accuracy-sparsity synergy.  By relaxing the penalty on the active set, Renet fundamentally alters the calculus of the ``One-Standard-Error'' (1-SE) rule. Typically, researchers must give up predictive accuracy in exchange for a sparser, more interpretable model. Renet dramatically reduces this traditional sparsity tax. As demonstrated in Proposition \ref{prop:Hierarchy} and validated across our benchmarks (e.g., the \textit{Student} and \textit{Crime} datasets), Renet allows for ultra-sparse models that retain (and sometimes exceed) the accuracy of their denser counterparts.  The standard Elastic Net, lacking a relaxation mechanism,  cannot replicate these results, for whenever the 1-SE rule is applied the resulting model is less accurate and / or less sparse. Without the 1-SE rule, Renet consistently improves upon the standard Elastic Net solution, whether in terms of accuracy, sparsity or, often, both. With the 1-SE rule, Renet's performance is comparable to the Adaptive Elastic Net (AEN) in well-behaved regimes where the latter excels, while remaining substantially more stable in less ideal scenarios, as our next point discusses.

Second, robustness in the $p \geq n$, low signal-to-noise ratio (SNR) and high-multicollinearity regimes. Unlike the AEN or  stepwise OLS, which become fragile or computationally intractable as dimensionality increases, Renet is engineered for stability. By enforcing strict convexity through the $\ell_2$ penalty, guaranteeing KKT compliance through dynamic objective selection (sub-path refitting or OLS interpolation) and imposing structural guardrails -- such as the zero-relaxation path for saturated models and the theoretically-grounded relaxation floor in ultra high-dimensional settings  -- Renet prevents the estimator from entering the high-variance region of the OLS solution path.
    
Lastly, computational efficiency. By using warm-starts during sub-path refits, and implementing an adaptive solver backend that dynamically dispatches optimization tasks to either \texttt{scikit-learn} or \texttt{Celer} based on the topology of the data, Renet often matches or beats the execution time of state-of-the-art solvers, despite its greater statistical sophistication.

In summary, Renet offers a general-purpose upgrade to the standard Elastic Net that addresses the specific limitations of its predecessors: it maintains the stability that AEN lacks in low-SNR, high-dimensional regimes, while providing a level of computational scalability that is fundamentally unavailable to stepwise OLS. By decoupling selection from estimation, Renet uniquely harnesses the 1-SE rule to achieve extreme sparsity without the predictive penalty typically incurred by heavy regularization.

\bigskip

\bibliographystyle{plainnat}
\bibliography{references}

\clearpage
\appendix

\section{Renet Implementation Pseudocode}
\label{ap:Algo}

\begin{algorithm}
\caption{Renet}
\label{alg:renet}
\begin{algorithmic}[1]
\State \textbf{Input:} $\mathbf{X} \in \mathbb{R}^{n \times p}$, $\mathbf{y} \in \mathbb{R}^n$, $\ell_1$-ratio $\alpha$, $\lambda$-grid $\Lambda$, $\theta$-grid $\Theta$, Folds $K$.
\State \textbf{Step 1: Joint Cross-Validation}
\For{each fold $k \in \{1, \dots, K\}$}
    \State $\hat{\mathbf{B}}_{EN} \gets$ Solve path on $\mathbf{X}_{train}$ for all $\lambda \in \Lambda$
    \State \If{$p \gg n_{train}$} $\theta_{floor} \gets \min(1.0, \log p / (2\sqrt{n_{train}}))$ \Comment{UHD complexity safeguard} \EndIf
    \For{each $\lambda_i \in \Lambda$}
        \State Identify active set $\mathcal{A}_{i} = \{j : |\hat{\beta}_{EN,j}(\lambda_i)| > 0 \}$
        \State $p_{active} \gets |\mathcal{A}_{i}|$
        \For{each $\theta_j \in \Theta$}
            \State $\theta_{eff} \gets \max(\theta_j, \theta_{floor})$ \If{$p_{active} \ge n_{train}$} $\theta_{eff} \gets 1.0$ \Comment{Saturation regime} \EndIf
            \State $\hat{\beta}_{Renet} \gets$ \Call{RelaxSolver}{$\mathbf{X}_{\mathcal{A}_i}, \mathbf{y}_{train}, \lambda_i, \theta_{eff}$}
            \State $MSE(i, j, k) \gets \|\mathbf{y}_{test} - \mathbf{X}_{test}\hat{\beta}_{Renet}\|_2^2$
        \EndFor
    \EndFor
\EndFor
\State $(\lambda^*, \theta^*) \gets$ Select best pair from $\text{avg}(MSE)$ \Comment{Optionally apply 1-SE rule here}

\State \textbf{Step 2: Final Model Estimation}
\State $\hat{\mathbf{B}}_{full} \gets$ Solve path on full data for all $\lambda \in \Lambda$
\State $\mathcal{A}^* \gets$ Active set from $\hat{\beta}_{full}(\lambda^*)$, $\theta_{floor} \gets \min(1.0, \log p / (2\sqrt{n}))$, $\theta_{final} \gets \max(\theta^*, \theta_{floor})$
\If{$|\mathcal{A}^*| \ge n$} $\theta_{final} \gets 1.0$ \EndIf
\State \Return $\hat{\beta}_{final} \gets$ \Call{RelaxSolver}{$\mathbf{X}_{\mathcal{A}^*}, \mathbf{y}, \lambda^*, \theta_{final}$}

\State
\Function{RelaxSolver}{$\mathbf{X}_{\mathcal{A}}, \mathbf{y}, \lambda, \theta$}
    \If{$\theta = 1.0$} \Return $\hat{\beta}_{EN}(\lambda)$ \EndIf
    \State $\hat{\beta}_{OLS} \gets (\mathbf{X}_{\mathcal{A}}^T \mathbf{X}_{\mathcal{A}})^{-1} \mathbf{X}_{\mathcal{A}}^T \mathbf{y}$ 
    \If{$\text{sgn}(\hat{\beta}_{EN}) = \text{sgn}(\hat{\beta}_{OLS})$}
        \State \Return $\theta \hat{\beta}_{EN} + (1-\theta) \hat{\beta}_{OLS}$ \Comment{Efficient convex blending}
    \Else \Comment{Sub-path refit}
    		\If{$p_{active} \gg n$} \Return Optimize Renet Objective with Celer on $\mathbf{X}_{\mathcal{A}}$ with $\lambda_{relax} = \theta \lambda$
    		\Else
        		\State \Return Optimize Renet Objective with vanilla CD on $\mathbf{X}_{\mathcal{A}}$ with $\lambda_{relax} = \theta \lambda$
        \EndIf
    \EndIf
\EndFunction
\end{algorithmic}
\end{algorithm}

\pagebreak

\section{Adaptive Elastic Net Implementation Pseudocode}
\label{ap:Algo_AEN}

\begin{algorithm}
\caption{Adaptive Elastic Net (AEN)}
\label{alg:adaptive_net}
\begin{algorithmic}[1]
\State \textbf{Input:} Data matrix $X \in \mathbb{R}^{n \times p}$, Target vector $y \in \mathbb{R}^n$. Parameters: $\gamma$ (adaptive power), l1\_ratio,  $K$ (cv folds).

\State \textbf{Step 1: Pilot Estimation (Ridge)}
\State Check and standardize $X, y$
\State $\hat{\beta}^{\text{ridge}} \gets \text{RidgeCV}(X, y, \text{cv}=K)$ \Comment{Obtain stable pilot estimates}

\State \textbf{Step 2: Compute Adaptive Weights}
\For{$j = 1$ to $p$}
    \State $w_j \gets \left( |\hat{\beta}^{\text{ridge}}_j|^\gamma + \epsilon_{\text{tol}} \right)^{-1}$ \Comment{Invert pilot coefs (e.g., $\epsilon_{\text{tol}}=10^{-12}$)}
\EndFor

\State \textbf{Step 3: Transform Feature Space}
\State $\tilde{X} \gets X$
\For{$j = 1$ to $p$}
    \State $\tilde{X}_{:,j} \gets X_{:,j} / w_j$ \Comment{Scale features by inverse weights}
\EndFor

\State \textbf{Step 4: Weighted Optimization (ElasticNet)}
\State $\hat{\beta}^{\text{raw}} \gets \text{ElasticNetCV}(\tilde{X}, y, \text{l1\_ratio}, \text{cv}=K)$

\State \textbf{Step 5: Rescale Coefficients}
\For{$j = 1$ to $p$}
    \State $\hat{\beta}_j \gets \hat{\beta}^{\text{raw}}_j / w_j$ \Comment{Transform back to original scale}
\EndFor
\State $\hat{\beta}_0 \gets \hat{\beta}^{\text{raw}}_0$ \Comment{Intercept remains unchanged}

\State \Return  $\hat{\beta}$
\end{algorithmic}
\end{algorithm}

\pagebreak

\section{Extended Results}
\label{ap:ExRes}

\begin{table}[ht]
\centering
\caption{Real-world out-of-sample results with standard errors in parenthesis. Best metrics shown in bold. DFN means the model ``Did Not Finish" within a 30-minute limit.}
\label{tab:appendix_results_Real}
\small
\begin{tabular}{@{}l l c c c c@{}}
\toprule
\textbf{Dataset} & \textbf{Metric} & \textbf{Stepwise OLS} & \textbf{Celer EN} & \textbf{Celer EN, 1-SE} & \textbf{Renet w/o relaxation} \\ 
\midrule
\multirow{3}{*}{Abalone} & $R^2$ & 0.519 (0.010) & \textbf{0.525 (0.010)} & 0.510 (0.010) & \textbf{0.525 (0.010)} \\
 & Num. Coeff.  & \textbf{6.1 (0.1)} & 7.3 (0.1) & 6.2 (0.2) & 7.3 (0.1) \\
 & Time (s)  & 0.154 & 0.130 & 0.131 & \textbf{0.016} \\ \midrule

\multirow{3}{*}{BMI} & $R^2$ & \textbf{0.990 (0.001)} & 0.989 (0.001) & 0.988 (0.001) & 0.989 (0.001) \\
 & Num. Coeff.  & \textbf{4.4 (0.5)} & 6.4 (0.3) & 4.6 (0.2) & 6.5 (0.3) \\
 & Time (s)  & 0.108 & 0.005 & 0.005 & \textbf{0.004} \\ \midrule

\multirow{3}{*}{College} & $R^2$ & \textbf{0.914 (0.007)} & 0.893 (0.009) & 0.858 (0.009) & 0.893 (0.009) \\
 & Num. Coeff.  & \textbf{9.7 (0.3)} & 16.7 (0.1) & 16.9 (0.1) & 16.7 (0.1) \\
 & Time (s)  & 0.335 & 0.014 & 0.015 & \textbf{0.005} \\ \midrule

\multirow{3}{*}{Crime} & $R^2$ & 0.656 (0.010) & \textbf{0.664 (0.010)} & 0.646 (0.010) & \textbf{0.664 (0.009)} \\
 & Num. Coeff.  & 49.4 (1.7) & 65.9 (2.7) & \textbf{13.9 (0.6)} & 65.7 (2.8) \\
 & Time (s)  & 86.974 & 0.905 & 0.922 & \textbf{0.063} \\ \midrule

\multirow{3}{*}{Diabetes} & $R^2$ & 0.469 (0.022) & \textbf{0.473 (0.021)} & 0.446 (0.018) & \textbf{0.473 (0.021)} \\
 & Num. Coeff.  & \textbf{6.1 (0.3)} & 8.8 (0.2) & 6.8 (0.2) & 8.8 (0.2) \\
 & Time (s)  & 0.096 & 0.014 & 0.015 & \textbf{0.006} \\ \midrule

\multirow{3}{*}{Doctor} & $R^2$ & \textbf{0.221 (0.013)} & 0.220 (0.012) & 0.184 (0.007) & 0.220 (0.012) \\
 & Num. Coeff.  & 7.3 (0.3) & 9.6 (1.0) & \textbf{2.8 (0.1)} & 9.4 (0.8) \\
 & Time (s)  & 0.508 & 0.067 & 0.063 & \textbf{0.007} \\ \midrule

\multirow{3}{*}{Parkinson} & $R^2$ & 0.166 (0.005) & \textbf{0.168 (0.005)} & 0.159 (0.005) & \textbf{0.168 (0.005)} \\
 & Num. Coeff.  & 14.3 (0.4) & 16.5 (0.3) & \textbf{12.3 (0.7)} & 16.5 (0.3) \\
 & Time (s)  & 1.199 & 1.049 & 1.052 & \textbf{0.054} \\ \midrule

\multirow{3}{*}{Riboflavin} & $R^2$ & DNF  & \textbf{0.608 (0.050)} & 0.607 (0.041) & 0.601 (0.055) \\
 & Num. Coeff.  & DNF  & 46.7 (2.1) & \textbf{25.1 (1.2)} & 49.7 (2.5) \\
 & Time (s)  & DNF & 0.320 & \textbf{0.282} & 0.307 \\ \midrule

\multirow{3}{*}{Student} & $R^2$ & 0.837 (0.014) & \textbf{0.843 (0.014)} & 0.816 (0.014) & \textbf{0.843 (0.014)} \\
 & Num. Coeff.  & 10.6 (0.6) & 8.2 (1.5) & \textbf{2.0 (0.0)} & 8.2 (1.5) \\
 & Time (s)  & 1.183 & 0.018 & 0.019 & \textbf{0.005} \\ \midrule

\multirow{3}{*}{Wine} & $R^2$ & \textbf{0.287 (0.005)} & \textbf{0.287 (0.005)} & 0.273 (0.005) & \textbf{0.287 (0.005)} \\
 & Num. Coeff.  & 10.2 (0.2) & 11.0 (0.0) & \textbf{6.4 (0.5)} & 11.0 (0.0) \\
 & Time (s)  & 0.352 & 0.158 & 0.155 & \textbf{0.010} \\
\bottomrule
\end{tabular}
\end{table}

\begin{table}[ht]
\centering
\caption{Synthetic out-of-sample results with standard errors in parenthesis. Best metrics shown in bold. DFN means the model ``Did Not Finish" within a 30-minute limit.}
\label{tab:appendix_results_Synth}
\small
\begin{tabular}{@{}l l c c c c@{}}
\toprule
\textbf{Dataset} & \textbf{Metric} & \textbf{Stepwise OLS} & \textbf{Celer EN} & \textbf{Celer EN, 1-SE} & \textbf{Renet w/o relaxation} \\ 
\midrule
\multirow{3}{*}{S1} & $R^2$ & \textbf{0.933 (0.000)} & \textbf{0.933 (0.000)} & 0.932 (0.000) & \textbf{0.933 (0.000)} \\
 & Num. Coeff.  & 12.0 (0.2) & 15.1 (0.3) & \textbf{10.0 (0.0)} & 15.1 (0.3) \\
 & Time (s)  & 12.248 & 2.786 & 2.825 & \textbf{0.061} \\ \midrule

\multirow{3}{*}{S2} & $R^2$ & \textbf{0.862 (0.002)} & \textbf{0.862 (0.002)} & 0.860 (0.002) & \textbf{0.862 (0.002)} \\
 & Num. Coeff.  & 11.3 (0.4) & 17.3 (0.5) & \textbf{6.2 (0.3)} & 17.4 (0.5) \\
 & Time (s)  & 1.142 & 0.078 & 0.080 & \textbf{0.009} \\ \midrule

\multirow{3}{*}{S3} & $R^2$ & 0.393 (0.011) & \textbf{0.395 (0.011)} & 0.381 (0.008) & \textbf{0.395 (0.011)} \\
 & Num. Coeff.  & 5.7 (0.4) & 5.0 (0.6) & \textbf{1.7 (0.2)} & 5.0 (0.6) \\
 & Time (s)  & 0.616 & 0.065 & 0.058 & \textbf{0.007} \\ \midrule

\multirow{3}{*}{S4} & $R^2$ & 0.882 (0.004) & \textbf{0.884 (0.003)} & 0.881 (0.003) & \textbf{0.884 (0.003)} \\
 & Num. Coeff.  & 26.4 (1.2) & 24.5 (0.7) & \textbf{10.2 (0.5)} & 24.6 (0.7) \\
 & Time (s)  & 50.983 & 0.066 & 0.068 & \textbf{0.012} \\ \midrule

\multirow{3}{*}{S5} & $R^2$ & \textbf{0.895 (0.002)} & \textbf{0.895 (0.002)} & 0.893 (0.002) & \textbf{0.895 (0.002)} \\
 & Num. Coeff.  & \textbf{82.2 (0.5)} & 96.2 (0.5) & 84.3 (0.4) & 96.1 (0.5) \\
 & Time (s)  & 147.269 & 3.182 & 3.303 & \textbf{0.425} \\ \midrule

\multirow{3}{*}{S6} & $R^2$ & 0.757 (0.015) & \textbf{0.760 (0.015)} & 0.752 (0.014) & \textbf{0.760 (0.015)} \\
 & Num. Coeff.  & 5.7 (0.3) & 6.4 (0.4) & \textbf{2.9 (0.1)} & 6.4 (0.4) \\
 & Time (s)  & 0.385 & 0.009 & 0.009 & \textbf{0.005} \\ \midrule

\multirow{3}{*}{S7} & $R^2$ & 0.698 (0.020) & \textbf{0.859 (0.006)} & 0.852 (0.006) & \textbf{0.859 (0.006)} \\
 & Num. Coeff.  & 98.1 (2.3) & 37.5 (1.6) & \textbf{22.7 (0.7)} & 37.9 (1.5) \\
 & Time (s)  & 1075.221 & 2.896 & \textbf{2.883} & 2.927 \\ \midrule

\multirow{3}{*}{S8} & $R^2$ & DNF & \textbf{0.141 (0.065)} & 0.105 (0.068) & 0.080 (0.099) \\
 & Num. Coeff.  & DNF  & 38.5 (3.0) & \textbf{16.5 (1.8)} & 45.1 (4.2) \\
 & Time (s)  & DNF & 0.203 & \textbf{0.199} & \textbf{0.199} \\ \midrule

\multirow{3}{*}{S9} & $R^2$ & 0.850 (0.020) & \textbf{0.904 (0.014)} & 0.895 (0.015) & \textbf{0.904 (0.015)} \\
 & Num. Coeff.  & 76.0 (1.8) & 69.0 (6.2) & \textbf{33.0 (4.7)} & 70.2 (5.7) \\
 & Time (s)  & 188.564 & 1.682 & 1.689 & \textbf{1.678} \\ \midrule

\multirow{3}{*}{S10} & $R^2$ & DNF  & \textbf{0.756 (0.012)} & 0.744 (0.010) & \textbf{0.756 (0.012)} \\
 & Num. Coeff.  & DNF  & 8.5 (0.6) & \textbf{5.8 (0.2)} & 8.5 (0.6) \\
 & Time (s)  & DNF & \textbf{1.158} & 1.170 & 1.245 \\
\bottomrule
\end{tabular}
\end{table}

\end{document}